\documentclass[a4paper]{article}
\usepackage{bm}
\usepackage{amsfonts}
\usepackage[english]{inputenc}
\usepackage{amsmath}
\usepackage{amssymb}
\usepackage{amsthm}
\usepackage{authblk}
\usepackage{color}

\begin{document}

\newtheorem{definition}{Definiton}[section]
\newtheorem{theorem}{Theorem}[section]

\title{Uniqueness of the static Einstein-Maxwell spacetimes with a photon sphere}

\author[1,2]{Stoytcho Yazadjiev\thanks{yazad@phys.uni-sofia.bg}}
\author[1]{Boian Lazov\thanks{boian\_lazov@yahoo.com}}
\affil[1]{Department of Theoretical Physics, Faculty of Physics, Sofia University, Sofia 1164, Bulgaria}
\affil[2]{Theoretical Astrophysics, Eberhard-Karls University of T\"ubingen, T\"ubingen 72076, Germany}

\maketitle

\begin{abstract}
We consider the problem of uniqueness of static and asymptotically flat Einstein-Maxwell spacetimes with a photon sphere $P^3$. We are using a naturally modified definition of a photon sphere for electrically charged spacetimes with the additional property that the one-form $\iota_{\xi}F$ is normal to the photon sphere. For simplicity we are restricting ourselves to the case of zero magnetic charge and assume that the lapse function regularly foliates the spacetime outside the photon sphere. With this information we prove that $P^3$ has constant mean curvature and constant scalar curvature. We also derive a few equations which we later use to prove the main uniqueness theorem, i. e. the static asymptotically flat Einstein-Maxwell spacetimes with a non-extremal photon sphere are isometric to the Reissner-Nordstr\"om one with mass $M$ and electric charge $Q$ subject to $\frac{Q^2}{M^2}\le \frac{9}{8}$.
\end{abstract}

\section{Introduction}

\indent Photon spheres are a well-known prediction of General Relativity and the generalised theories of gravitation. They are regions of spacetime where light can be confined to closed orbits. Photon spheres are expected to be a characteristic of ultracompact objects such as black holes, neutron stars, wormholes and naked singularities \cite{Iyer1985}-\cite{Baldiotti2014}. They are closely related to gravitational lensing and thus play an important role in astronomy and astrophysics. From an astrophisical perspective a photon sphere is a timelike hypersurface on which the light bending angle is unboundedly large \cite{Virbhadra2000},\cite{Virbhadra2002}. The presence of a photon sphere around an ultracompact object leads to the so-called relativistic images, which are vitally important for observational astrophysics \cite{Virbhadra2000},\cite{Virbhadra2002}.\\
\indent Photon spheres are objects with very specific characteristics. One of them is that the lapse function in static spacetimes is constant on the photon sphere. Another is that  the photon sphere is a totally umbilic hypersurface with constant mean and scalar curvatures as well as constant surface gravity
\cite{Cederbaum2014}, \cite{Yazadjiev2015}. These properties make them similar to event horizons. This naturally leads to the question whether photon spheres can be used in the classification of solutions to a given gravitational theory. This has indeed been done for the Schwarzschild spacetime. Namely it has been proven recently that static asymptopically flat solutions to the Einstein equations in vacuum with mass $M$ possessing a photon sphere are isometric to the Schwarzschild solution \cite{Cederbaum2014}. A similar uniqueness theorem has also been proven for the static and asymptotically flat solutions to the Einstein-scalar field equations \cite{Yazadjiev2015}. In general this uniqueness question is harder than the one for black hole horizons \cite{Heusler1996}, because the class of spacetimes with a photon sphere is much larger than the one with an event horizon.\\
\indent In our paper we consider static and asymptotically flat vacuum solutions to the Einstein-Maxwell system of equations containing a photon sphere. We restrict ourselves to the case of zero magnetic charge for simplicity and we define the photon sphere to have one more property compared to the vacuum case, namely the electric field to be  orthogonal to the photon sphere. We prove that these spacetimes with $M^2\ne Q^2$ are isometric to the Reissner-Nordstr\"om solution with mass $M$ and electric charge $Q$ subject to the condition
$\frac{Q^2}{M^2}\le \frac{9}{8}$.

\section{Preliminary definitions and equations}

\indent In the present paper we consider Einstein-Maxwell gravity described by the following action:
\begin{align}
\mathcal{S}=\frac{1}{16\pi}\int_{\mathfrak{L}^4}\mathrm{d}^4 x \sqrt{- \mathfrak{g}} (\mathfrak{R} - F_{\mu\nu}F^{\mu\nu}),
\end{align}

\noindent where the spacetime manifold is denoted by $(\mathfrak{L}^4,\mathfrak{g})$, $\mathfrak{R}$ is the Ricci scalar curvature and $F$ is the Maxwell tensor. From this action we get the field equations:
\begin{align}
&\mathfrak{R}_{\mu\nu}=2\left( F_{\mu\alpha}F^{\ \alpha}_{\nu} -\frac{1}{4}\mathfrak{g}_{\mu\nu}F_{\alpha\beta}F^{\alpha\beta} \right),\\
&\mathrm{d}\star F=0,\\
&\mathrm{d}F=0,
\end{align}
where $\star$ is the Hodge star operator.

\indent We are considering static spacetimes, which means that there exists a smooth Riemannian manifold $(M^3,g)$ and a smooth lapse function
$N:M^3\longrightarrow \mathbb{R}^+$  such that
\begin{align}
\mathfrak{L}^4=\mathbb{R}\times M^3,\ \mathfrak{g}=-N^2\mathrm{d}t^2+g.\label{stmetr}
\end{align}

Let $\xi=\frac{\partial}{\partial t}$ be the timelike Killing vector. Then we define staticity of the Maxwell field,
\begin{align}
\mathcal{L}_{\xi}F=0.
\end{align}
We shall focus on  the purely electric case with $\iota_{\xi}\star F=0$.

\indent The spacetimes we are considering are also asymtotically flat. Asymptotic flatness is defined in the usual way. The spacetime is asymptotically
flat if there exists a compact set $K\in M^3$ such that $M^3 \setminus K$ is diffeomorphic to $\mathbb{R}^3\setminus{\bar B}$ where ${\bar B}$ is the closed unit ball centered at the origin in $\mathbb{R}^3$ and such that
\begin{eqnarray}
g=\delta + O(r^{-1}), \;\;\; N= 1 - \frac{M}{r} + O(r^{-2}),
\end{eqnarray}
with respect to the standard radial coordinate on $\mathbb{R}^3$.  The asymptotic expansion of the electromagnetic field can be
derived easily and we have
\begin{eqnarray}
 F= -\frac{Q}{r^2} dt\wedge dr + O(r^{-3}).
\end{eqnarray}\

\noindent As usual $M$ and $Q$ are the mass and the electric charge, respectively. We will focus our study on the physically interesting case with $M>0$ and $Q\ne0$.\\
\indent We also need to define the photon sphere. We start with the definition of a photon surface \cite{Claudel2001}.
\begin{definition}
An embedded timelike hypersurface $(P^3,p)\hookrightarrow (\mathfrak{L}^4,\mathfrak{g})$ is called a photon surface if any null geodesic initially tangent to $P^3$ remains tangent to $P^3$ as long as it exists.
\end{definition}

\noindent Next we define a photon sphere.
\begin{definition}
Let $(P^3,p)\hookrightarrow (\mathfrak{L}^4,\mathfrak{g})$ be a photon surface. Then $P^3$ is called a photon sphere if the lapse function $N$ is
constant on $P^3$ and the one-form $\iota_{\xi}F$ is  normal to $P^3$.
\end{definition}

We will make an additional technical assumption. We will assume that the lapse function $N$ regularly foliates the spacetime outside the photon sphere, i. e.
\begin{align}
\rho^{-2}=g\left( {}^g\!\nabla N,{}^g\!\nabla N \right)\ne 0
\end{align}

\noindent outside the photon sphere. The spatial part of this exterior region will be denoted by $M^3_{\mathrm{ext}}$ and by definition it has as inner boundary the intersection $\Sigma$ of the outermost photon sphere with the time slice $M^3$. By definition $\Sigma$ is given by $N=N_0$ for some $N_0\in \mathbb{R}^+$. As a consequence of our assumption
all level sets $N=const$, including $\Sigma$, are topological spheres and $M^3_{\mathrm{ext}}$ is topologically ${\mathbb{S}}^2\times \mathbb{R}$.

\indent We  define the electric field one form $E$ by
\begin{align}
E=-\iota_{\xi}F,
\end{align}
and it satisfies  $\mathrm{d}E=0$ as a consequence of the field equations and the electromagnetic staticity. Since $M^3_{\mathrm{ext}}$ is simply connected  this implies the existence of an electric potential $\Phi$ such that $E=\mathrm{d}\Phi$. Using the electric field one{\color{red}-}form we can write an explicit expression for $F$,
\begin{align}
F=-N^{-2}\xi \wedge \mathrm{d}\Phi.
\end{align}
By definition the electric field $E$ is normal to the photon sphere and therefore the electrostatic potential $\Phi$ is constant on the photon sphere. Since the electrostatic potential is defined up to a constant, without loss of generality we shall set $\Phi_{\infty}=0$.\\
\indent Using the form of the metric (\ref{stmetr}) and the form of the Maxwell tensor  we can obtain the dimensionally reduced static Einstein-Maxwell field equations:
\begin{align}
&{}^g\!\Delta N=N^{-1}\,{}^g\!\nabla^i\Phi\,{}^g\!\nabla_i\Phi,\label{EM1}\\
&{}^g\!R_{ij}=N^{-1}\,{}^g\!\nabla_i\,{}^g\!\nabla_jN+N^{-2}\left( g_{ij}\,{}^g\!\nabla^k\Phi\,{}^g\!\nabla_k\Phi -2\,{}^g\!\nabla_i\Phi\,{}^g\!\nabla_j\Phi \right),\label{EM2}\\
&{}^g\!\nabla^i\left( N^{-1}\,{}^g\!\nabla_i\Phi \right)=0.\label{EM3}
\end{align}

In \cite{Israel1968} Israel derived a divergence identity for the lapse function and the electrostatic potential which in our notation is given by

\begin{eqnarray}
{}^g\!\nabla^i\left[2\Phi{}^g\!\nabla_iN - (N + N^{-1}\Phi^2){}^g\!\nabla_i\Phi\right]=0.
\end{eqnarray}
The Gauss theorem applied to this identity  leads to the following functional dependence between the lapse function $N_{0}$ and the electrostatic potential
$\Phi_{0}$ on $\Sigma$
\begin{eqnarray}\label{FDNPHI0}
N^2_{0}=\Phi^2_{0} - 2\frac{M}{Q}\Phi_{0} + 1,
\end{eqnarray}
where we have taken into account that $\Phi$ and $N$ are constant on $\Sigma$ and that $\Phi_{\infty}=0$. Moreover, this functional dependence holds
not only on $\Sigma$ but also  on the whole $M^3_{\mathrm{ext}}$, namely
\begin{eqnarray}\label{FDNPHI}
N^2= \Phi^2 - 2\frac{M}{Q}\Phi + 1.
\end{eqnarray}

In order to prove this we can use the following divergence identity:
\begin{eqnarray}
N w_{i}w^{i}= \frac{1}{2}  {}^g\!\nabla^i\left[(-N^2 + \Phi^2 - 2\frac{M}{Q}\Phi + 1)w_{i} \right],
\end{eqnarray}
\noindent where $w_{i}$ is given by

\begin{eqnarray}
w_{i}=- {}^g\!\nabla_i N + N^{-1}\left(\Phi - \frac{M}{Q}\right) {}^g\!\nabla_i \Phi .
\end{eqnarray}

\noindent This identity is a consequence of the dimensionally reduced field equations. Integrating this identity on $M^3_{\mathrm{ext}}$ with the help of the Gauss theorem and taking into account (\ref{FDNPHI0}) and the asymptotic behavior of $N$ and $\Phi$ we obtain that $w_{i}=0$ on $M^3_{\mathrm{ext}}$. Hence we find
$N^2=\Phi^2 - 2\frac{M}{Q}\Phi + C$ with $C$ being a constant. From the asymptotic behavior of $N$ and $\Phi$  we conclude that $C=1$ which proves (\ref{FDNPHI}).\\
By the maximum principle for elliptic partial differential equations  and by the asymptotic behavior of $N$ for $r\to \infty$ we obtain for the values of $N$ on $M^3_{\mathrm{ext}}$ the following inequality:
\begin{align}
N_0\le N<1.
\end{align}

\section{Auxiliary equations and theorems}

\indent For an isometric embedding $(A^n,a)\hookrightarrow(B^{n+1},b)$ with unit normal $\eta$ and second fundamental form $II(X,Y)=b({}^b\!\nabla_X\eta,Y)$, $X,Y\in \Gamma(TA^n)$, there are a few other equations we will use.
\begin{itemize}
\item The Codazzi equation, which reads
\begin{align}
b({}^b\!R(X,Y,\eta),Z)=({}^a\!\nabla_XII)(Y,Z)-({}^a\!\nabla_YII)(X,Z)\label{codd}
\end{align}
\noindent for all $X,Y,Z\in \Gamma(TA^n)$.
\item The contracted Gauss equation, which reads
\begin{align}
{}^b\!R-2\tau\,{}^b\!R(\eta,\eta)={}^a\!R-\tau({}^a\!\mathrm{tr}II)^2+\tau\left|II\right|^2,\label{cgauss}
\end{align}
\noindent where $\tau=b(\eta,\eta)$.
\item An equation satisfied for every smooth function $f:B^{n+1}\longrightarrow \mathbb{R}$, if $\tau=1$, which reads
\begin{align}
{}^b\!\Delta f={}^a\!\Delta f + {}^b\!\nabla^2f(\eta,\eta)+({}^a\!\mathrm{tr}II)\eta(f).\label{functe}
\end{align}
\end{itemize}

\indent The results in this section are not only vital for the proof of our main theorem but are also interesting  for their own. First we recall a result by Claudel, Virbhadra and Ellis \cite{Claudel2001}.
\begin{theorem}
Let $(P^3,p)\hookrightarrow(\mathfrak{L}^4,\mathfrak{g})$ be an embedded timelike hypersurface. Then $P^3$ is a photon surface if and only if it is totally umbilic (i. e. its second fundamental form is pure trace).
\end{theorem}

\noindent The second fundamental form of $P^3$ can then be written as $\mathfrak{h}=\frac{\mathfrak{H}}{3}p$, where $\mathfrak{H}$ is the mean curvature of $P^3$. Using this we can prove a theorem concerning the mean and scalar curvatures of the photon sphere $P^3$ in our spacetime.
\begin{theorem}
Let $(\mathfrak{L}^4,\mathfrak{g},F)$ be a static, asymptotically flat Einstein-Maxwell spacetime possessing a photon sphere $(P^3,p)\hookrightarrow(\mathfrak{L}^4,\mathfrak{g})$.  Then $P^3$ has constant mean curvature (CMC) and constant scalar curvature (CSC).
\end{theorem}

\begin{proof}
To prove the theorem we will first use the Codazzi equation (\ref{codd}) for $(P^3,p)\hookrightarrow(\mathfrak{L}^4,\mathfrak{g})$ with unit normal $\nu$:
\begin{align}
\mathfrak{g}(\mathfrak{R}(X,Y,\nu),Z)=&({}^p\!\nabla_X\mathfrak{h})(Y,Z)-({}^p\!\nabla_Y\mathfrak{h})(X,Z)\\
=&X\left(\frac{\mathfrak{H}}{3}\right)p(Y,Z)-Y\left(\frac{\mathfrak{H}}{3}\right)p(X,Z),\notag
\end{align}

\noindent where $X,Y,Z\in \Gamma(TP^3)$. Now we contract the slots $X$ and $Z$ and obtain
\begin{align}
\mathfrak{R}(Y,\nu)=(1-3)Y\left(\frac{\mathfrak{H}}{3}\right).\label{cricc}
\end{align}

\noindent Using the Einstein-Maxwell equations to calculate the left-hand side of (\ref{cricc}) we get
\begin{align}
\mathfrak{R}_{\alpha\beta}Y^{\alpha}\nu^{\beta}=2\frac{\iota_EE}{(\iota_{\xi}\xi)^2}\iota_{\xi}Y\iota_{\xi}\nu+2\frac{1}{\iota_{\xi}\xi}\iota_EY\iota_E\nu.
\end{align}

\noindent Hence, taking into account that  $E^{\mu}$ is normal to $P^3$ and that $\xi\in \Gamma(TP^3)$,  it follows that $\iota_{E}Y=0$ and $\iota_{\xi}\nu=0$. Finally, we obtain
\begin{align}
0=(1-3)Y\left(\frac{\mathfrak{H}}{3}\right),
\end{align}

\noindent which means that $P^3$ has CMC since $Y$ is an arbitrary tangent vector to $P^3$.\\
\indent Next we need to prove that $P^3$ has constant scalar curvature. To do this we will use the contracted Gauss equation (\ref{cgauss}).
\begin{align}
\mathfrak{R}-2\tau\mathfrak{R}(\nu,\nu)=&{}^p\!R-\tau({}^p\!\mathrm{tr}\mathfrak{h})^2+\tau\left|\mathfrak{h}\right|^2\\
=&{}^p\!R-\mathfrak{H}^2+\frac{\mathfrak{H}^2}{3}\notag\\
=&{}^p\!R-\frac{2}{3}\mathfrak{H}^2.\notag
\end{align}

\noindent From the field equations we know that the Ricci scalar $\mathfrak{R}$ vanishes, leading to
\begin{align}
{}^p\!R=\frac{2}{3}\mathfrak{H}^2-2\mathfrak{R}(\nu,\nu).
\end{align}

\noindent Here we can calculate $\mathfrak{R}(\nu,\nu)$, again using the field equations.
\begin{align}
\mathfrak{R}(\nu,\nu)=2\left( \frac{1}{\iota_{\xi}\xi}(\iota_E\nu)^2 - \frac{1}{2}\frac{1}{\iota_{\xi}\xi}\iota_EE \right).\label{ricgauss}
\end{align}

\noindent Since $E$ is normal to $P^3$ we have $\iota_E\nu=E_{\nu}$. Then (\ref{ricgauss}) takes the form
\begin{align}
\mathfrak{R}(\nu,\nu)=-\frac{E_{\nu}^2}{N^2}.
\end{align}

\indent Finally, we obtain that $P^3$ has CSC, given by the expression
\begin{align}
{}^p\!R=\frac{2}{3}\mathfrak{H}^2+2\frac{E_{\nu}^2}{N^2}. \label{riccc}
\end{align}

\indent We will show below that $E_{\nu}$ is constant on $P^3$ which shows that ${}^p\!R$ is constant. This completes the proof.
\end{proof}

\indent Next we will derive a few more useful relations. We start off by computing the second fundamental form $h$ of $(\Sigma,\sigma)\hookrightarrow(M^3,g)$ with unit normal $\nu$. Let $X,Y,Z\in \Gamma(T\Sigma)$ be arbitrary tangent vectors to $\Sigma$. Then
\begin{align}
h(X,Y)=g({}^g\!\nabla_X\nu,Y)=\mathfrak{g}({}^{\mathfrak{g}}\!\nabla_X\nu,Y)=\mathfrak{h}(X,Y)=\frac{\mathfrak{H}}{3}p(X,Y)=\frac{\mathfrak{H}}{3}\sigma(X,Y).
\end{align}

\noindent Thus we see that $\Sigma$ is totally umbilic and has CMC $H=\frac{2}{3}\mathfrak{H}$. Now we make use of the Codazzi equation for $(\Sigma,\sigma)\hookrightarrow(M^3,g)$.
\begin{align}
g({}^g\!R(X,Y,\nu),Z)=&({}^{\sigma}\!\nabla_Xh)(Y,Z)-({}^{\sigma}\!\nabla_Yh)(X,Z)\\
=&{}^{\sigma}\!\nabla_X\left(\frac{\mathfrak{H}}{3}\right)\sigma(Y,Z)-{}^{\sigma}\!\nabla_Y\left(\frac{\mathfrak{H}}{3}\right)\sigma(X,Z).\notag
\end{align}

\noindent After contracting the $X$ and $Z$ slots we obtain
\begin{align}
{}^g\!R(Y,\nu)=0,
\end{align}

\noindent due to $\mathfrak{H}$ being constant on $\Sigma$. We will use this result to prove that $\nu(N)$ is constant on $\Sigma$. For this purpose we calculate the Lie derivative, using the field equations,
\begin{align}
{}^g\!\mathcal{L}_X(\nu(N))=&X^a\,{}^g\!\nabla_a(\nu(N))=X^a\,{}^g\!\nabla_a\left( \rho\,{}^g\!\nabla_bN\,{}^g\!\nabla^bN \right)\\
=&\rho 2X^a\,{}^g\!\nabla^bN\,{}^g\!\nabla_a\,{}^g\!\nabla_bN=2({}^g\!\nabla^2N)(X,\nu)\notag\\
=&2N\,{}^g\!R(\nu,X)-2N^{-1}g(\nu,X){}^g\!\nabla^k\phi\,{}^g\!\nabla_k\phi+4N^{-1}\,{}^g\!\nabla_i\phi\nu^i\,{}^g\!\nabla_j\phi X^j\notag\\
=&2N\,{}^g\!R(\nu,X)=0.\notag
\end{align}
As a direct consequence of (\ref{FDNPHI}) we obtain that $E_{\nu}$ is also constant on $\Sigma$ and thus on $P^3$.

\indent Next we use equation (\ref{functe}), again for $(\Sigma,\sigma)\hookrightarrow(M^3,g)$,
\begin{align}
{}^g\!\Delta N={}^{\sigma}\!\Delta N +{}^g\!\nabla^2N(\nu,\nu)+({}^{\sigma}\!\mathrm{tr}h)\nu(N).\label{funs}
\end{align}

\noindent By the Einstein-Maxwell equations, (\ref{funs}) transforms into
\begin{align}
2N^{-1}\iota_EE=N\,{}^g\!R(\nu,\nu)+2N^{-1}(\iota_E\nu)^2+\nu(N)H.\label{imp1}
\end{align}

\noindent The contracted Gauss equation gives
\begin{align}
{}^g\!R-2\,{}^g\!R(\nu,\nu)={}^{\sigma}\!R-\frac{H^2}{2}.\label{imp2}
\end{align}

\noindent Finally, contracting the field equations we get
\begin{align}
{}^g\!R=2N^{-2}\iota_EE.\label{imp3}
\end{align}

\noindent We can combine (\ref{imp1}-\ref{imp3}) to get
\begin{align}
N\,{}^{\sigma}\!R=2N^{-1}E_{\nu}^2+2\nu(N)H+\frac{1}{2}NH^2.\label{imp4}
\end{align}

\indent We next integrate (\ref{imp4}) over $\Sigma$.
\begin{align}
\int_{\Sigma}N\,{}^{\sigma}\!R\mathrm{d}\mu=\int_{\Sigma}2N^{-1}E_{\nu}^2\mathrm{d}\mu+\int_{\Sigma}2\nu(N)H\mathrm{d}\mu +\int_{\Sigma}\frac{1}{2}NH^2\mathrm{d}\mu.\label{imp5}
\end{align}

\noindent Let us denote the area of $\Sigma$ by $A_{\Sigma}$. With this, using the Gauss-Bonnet theorem and noting that for two-dimensional manifolds $R=2K$, where $K$ is the Gaussian curvature, we can transform (\ref{imp5}) so that it becomes
\begin{align}
N_0=\frac{E_{\nu}^2 A_{\Sigma}}{N_0}+ \frac{1}{4\pi}[\nu(N)] H A_{\Sigma} + \frac{1}{16\pi} H^2 A_{\Sigma}.\label{prem1}
\end{align}

\indent Next we will use the Komar definition for the spacetime mass,
\begin{align}
M=&-\frac{1}{8\pi}\int_{S^2_{\infty}}{}^{\mathfrak{g}}\!\nabla^{\alpha}\xi^{\beta}\mathrm{d}S_{\alpha\beta}\\
=&-\frac{1}{8\pi}\int_{\Sigma}{}^{\mathfrak{g}}\!\nabla^{\alpha} \xi^{\beta}\mathrm{d}S_{\alpha\beta} -\frac{1}{4\pi}\int_{M^3_{\mathrm{ext}}}\mathfrak{R}^{\alpha}_{\beta} \xi^{\beta}\eta_{\alpha}\sqrt{g}\mathrm{d}^3y\notag
\end{align}

\noindent The first term here is the mass of the photon sphere $M_P$. Using the field equations we transform the second term and arrive at the expression
\begin{align}
M=M_P-\frac{1}{4\pi}\int_{\Sigma}\frac{\Phi}{N}E_{\nu}\mathrm{d}\mu.
\end{align}

\noindent Since the electric charge is given by $Q= -\frac{1}{4\pi}\int_{\Sigma}\star F$, one can show that the mass becomes
\begin{align}
M=M_P + Q\Phi_{0}.\label{prem2}
\end{align}

\noindent Similarly we can calculate explicitly the expression for the mass of the photon sphere, arriving at $M_P=\frac{1}{4\pi}[\nu(N)]A_{\Sigma}$. With this we rewrite (\ref{prem1}) to become
\begin{align}
N_0=\frac{1}{4\pi}\frac{E^2_{\nu}A_{\Sigma}}{N_0}+M_P H + \frac{1}{16\pi}{}N_0H^2 A_{\Sigma}.\label{prem4}
\end{align}

\indent Finally, we will use the contracted Gauss equation again, this time  for $(\Sigma^2,\sigma)\hookrightarrow(P^3,p)$ with a unit normal $\eta$,
\begin{align}
{}^p\!R+2\,{}^p\!R(\eta,\eta)={}^{\sigma}\!R.
\end{align}

\noindent Remembering (\ref{riccc}) and using  ${}^p\!R(\eta,\eta)=0$, we get
\begin{align}
{}^{\sigma}\!R={}^p\!R=\frac{2}{3}\mathfrak{H}^2+2\frac{E_{\nu}^2}{N^2}=\frac{3}{2}H^2 +2\frac{E_{\nu}^2}{N^2}. \label{riccsig}
\end{align}

\noindent We integrate (\ref{riccsig}) over $\Sigma$, again taking into account the Gauss-Bonnet theorem,
\begin{align}
1=\frac{3}{16\pi}H^2 A_{\Sigma} + \frac{1}{4\pi}\frac{E_{\nu}^2 A_{\Sigma}}{N_0^2}.\label{prem3}
\end{align}

\indent We can make use of the definition of the electric charge to write it in a form containing the function $E_{\nu}$, i. e. $Q=-\frac{A_{\Sigma} E_{\nu}}{4\pi N_0}$. Then after a few calculations involving (\ref{prem2}), (\ref{prem4}) and (\ref{prem3}) we arrive at
\begin{align}
1=\frac{4\pi Q^2}{A_{\Sigma}} + \frac{3}{2}\left( M - Q\Phi_{0}\right)H.
\end{align}

Another very useful relation that can be easily obtained from (\ref{prem4}) and (\ref{prem3}) is the following:
\begin{eqnarray}\label{NH}
2\nu(N)=N_{0}H .
\end{eqnarray}

\section{Uniqueness theorem}

We first define the notion of {\it non-extremal } photon sphere.
\begin{definition}
A photon sphere is called non-extremal if
\begin{eqnarray}
\frac{1}{4\pi} H^2 A_{\Sigma}\ne 1.
\end{eqnarray}
\end{definition}

\noindent In the case of Einstein-Maxwell equations, using the relations derived in the previous section, it is not difficult to show that the photon sphere is
non-extremal only if $M^2\ne Q^2$.\\
\indent The main result of the present paper is the following:
\begin{theorem}
Let $(\mathfrak{L}_{ext}^4,\mathfrak{g}, F)$ be a static  and asymptotically flat spacetime with  given mass $M$ and charge $Q$, satisfying the Einstein-Maxwell equations and possessing a non-extremal photon sphere as an inner boundary of $\mathfrak{L}_{ext}^4$. Assume that the lapse function regularly foliates $\mathfrak{L}_{ext}^4$. Then $(\mathfrak{L}_{ext}^4,\mathfrak{g}, F)$ is isometric to the Reissner-Nordstr\"om spacetime with mass $M$ and charge $Q$ subject to the inequality $\frac{Q^2}{M^2}\le \frac{9}{8}$.
\end{theorem}

\begin{proof}
In proving the theorem we shall follow  \cite{Yazadjiev2015} with some technical modifications due to the fact that the target space metric for
the Einstein-Maxwell equations is Lorentzian in contrast to \cite{Yazadjiev2015}, where the target space metric is Riemannian.

Let us consider the 3-metric $\gamma_{ij}$ on $M^3_{ext}$ defined by

\begin{eqnarray}
\gamma_{ij}= N^2 g_{ij}.
\end{eqnarray}

\noindent Rewriting the dimensionally reduced static Einstein-Maxwell equations in terms of the new metric $\gamma_{ij}$ we have

\begin{align}
&R(\gamma)_{ij}= 2D_{i}\ln(N)D_{j}\ln(N) - 2 N^{-2}D_{i}\Phi D_{j}\Phi, \nonumber \\
&D_{i}D^{i}\ln(N)=N^{-2}D_{i}\Phi D^{i}\Phi, \\
&D_{i}\left(N^{-2}D^{i}\Phi\right)=0. \nonumber
\end{align}

\noindent Further reduction can be achieved by taking in account that $N$ and $\Phi$ are functionally dependent via (\ref{FDNPHI}). However instead of using
$N$ or $\Phi$ it is convenient to use another potential ${\tilde \lambda}$ defined by

\begin{eqnarray}
\mathrm{d}{\tilde \lambda}= - N^{-2}\mathrm{d}\Phi, \;\;  \;\; {\tilde \lambda}_{\infty}=0.
\end{eqnarray}
In terms of this new potential the dimensionally reduced static Einstein-Maxwell equations
become

\begin{eqnarray}
&&R(\gamma)_{ij}= 2\left(\frac{M^2}{Q^2} -1\right) D_{i}{\tilde \lambda} D_{j}{\tilde \lambda}, \\
&&D_{i}D^{i}{\tilde \lambda}=0.
\end{eqnarray}

In fact the potentials $N$ and $\Phi$ parameterise the coset $SL(2,\mathbb{R})/SO(1,1)$ with a metric
\begin{eqnarray}
G_{AB}\mathrm{d}\phi^A \mathrm{d}\phi^B= N^{-2} \mathrm{d}N^2 - N^{-2}\mathrm{d}\Phi^2
\end{eqnarray}

\noindent and it is not difficult to see that $(N(\tilde \lambda), \Phi(\tilde \lambda))$ is a geodesic on  $SL(2,\mathbb{R})/SO(1,1)$ with

\begin{eqnarray}
G_{AB}\frac{\mathrm{d}\phi^A}{\mathrm{d}{\tilde \lambda}} \frac{\mathrm{d}\phi^B}{\mathrm{d}{\tilde \lambda}}= \frac{M^2}{Q^2} -1.
\end{eqnarray}
Depending on the ratio $\frac{Q^2}{M^2}$ we have three types of geodesics which we will formally call "spacelike" for $\frac{Q^2}{M^2}<1$,
"timelike" for $\frac{Q^2}{M^2}>1$ and "null" for $\frac{Q^2}{M^2}=1$. The null geodesics correspond to extremal photon spheres and that is why they will not be considered here. The other types of geodesics have to be considered separately.

\medskip
\noindent

{\bf Case $\frac{Q^2}{M^2}<1$.}

\medskip
\noindent

In studying the case of spacelike geodesics we shall use the "affine parameter" $\lambda$ given by

\begin{eqnarray}
\lambda=\sqrt{\frac{M^2}{Q^2} -1} \;{\tilde \lambda}.
\end{eqnarray}

\noindent We proceed further by  considering the inequalities \cite{Yazadjiev2015}

\begin{eqnarray}\label{II1}
\int_{M^{3}_{ext}} D^{i} \left[\Omega^{-1}\left(\Gamma D_{i}\chi - \chi D_{i}\Gamma\right) \right] \sqrt{\gamma} \mathrm{d}^3x \ge 0
\end{eqnarray}

\noindent and

\begin{eqnarray}\label{II2}
\int_{M^{3}_{ext}} D^{i} \left(\Omega^{-1} D_{i}\chi \right) \sqrt{\gamma} \mathrm{d}^3x \ge \int_{M^{3}_{ext}} D^{i} \left[\Omega^{-1}\left(\Gamma D_{i}\chi - \chi D_{i}\Gamma\right) \right] \sqrt{\gamma} \mathrm{d}^3x ,
\end{eqnarray}
\noindent where $\chi$, $\Gamma$ and $\Omega$ are defined by

\begin{eqnarray}
\chi= \left(\gamma^{ij}D_{i}\Gamma D_{j}\Gamma \right)^{\frac{1}{4}}, \;\;\; \Gamma= \frac{1- e^{2\lambda}}{1 + e^{2\lambda}}, \;\;\; \Omega=\frac{4e^{2\lambda}}{(1 + e^{2\lambda})^2}.
\end{eqnarray}
The equalities in (\ref{II1}) and (\ref{II2}) hold if and only if the Bach tensor $R(\gamma)_{ijk}$ vanishes \cite{Yazadjiev2015}.

After rather long and unpleasant calculations with the help  of Gauss theorem, and taking into account (\ref{NH}), one can show that the first inequality (\ref{II1}) is equivalent to

\begin{eqnarray}\label{INEQ1}
\Phi^2_{0} - \frac{3}{2}\frac{M}{Q}\Phi_{0} + \frac{1}{2}\le 0,
\end{eqnarray}
while the second inequality  (\ref{II2}) gives

\begin{eqnarray}\label{INEQ2}
\Phi^2_{0} - \frac{3}{2}\frac{M}{Q}\Phi_{0} + \frac{1}{2}\ge 0.
\end{eqnarray}

\noindent Hence we conclude that

\begin{eqnarray}\label{HSEPHI}
\Phi^2_{0} - \frac{3}{2}\frac{M}{Q}\Phi_{0} + \frac{1}{2}= 0
\end{eqnarray}
and therefore $R(\gamma)_{ijk}=0$. Since the Bach tensor vanishes we conclude that the metric $\gamma_{ij}$ is conformally flat. Obviously the same holds for the metric $g_{ij}$.

If we denote the surfaces of constant $N$ embedded in $M^3$ by $\Sigma_N$, $(\Sigma_N,\sigma)\hookrightarrow(M^3,g)$, we can write our spacetime metic in the form:
\begin{align}
\mathfrak{g}=-N^2\mathrm{d}t^2+\rho^2\mathrm{d}N^2+\sigma_{AB}\mathrm{d}x^A \mathrm{d}x^B.
\end{align}

Using the dimensionally reduced field equations one can show that

\begin{eqnarray}
R(g)_{ijk}R(g)^{ijk}= \frac{8}{N^4\rho^4}\left[ \frac{M^2 - Q^2}{M^2  - Q^2 +   Q^2 N^2}  \right]^2  \times \\\nonumber  \left[\left(h^{\Sigma_{N}}_{AB}- \frac{1}{2}H^{\Sigma_{N}} \sigma_{AB}\right)
\left(h^{\Sigma_{N}\, AB}- \frac{1}{2}H^{\Sigma_{N}} \sigma^{AB} \right)    + \frac{1}{2\rho^2}\sigma^{AB}\partial_{A}\rho\partial_{B}\rho \right] ,
\end{eqnarray}
where $h^{\Sigma_{N}}_{AB}$ is the second fundamental form of $\Sigma_{N}$ and $H^{\Sigma_{N}}$ is its trace. Then, for $\frac{Q^2}{M^2} <1$, we conclude
that

\begin{eqnarray}
h^{\Sigma_{N}}_{AB}= \frac{1}{2}H^{\Sigma_{N}} \sigma_{AB},\;\;\; \partial_{A}\rho=0.
\end{eqnarray}
The space geometry is therefore spherically symmetric. Then one can easily
show that the spacetime  is isometric to the Reissner-Nordstr\"om spacetime with $\frac{Q^2}{M^2}<1$. This can be done  by direct computation or by using Birkhoff's theorem for the Einstein-Maxwell equations.

Eq. (\ref{HSEPHI}) determines the value of the electrostatic potential and  the lapse function (via (\ref{FDNPHI})) on the photon sphere. In fact eq. (\ref{HSEPHI})
has two solutions{\color{red}:}

\begin{eqnarray}\label{SOLEPS}
\Phi^{\pm}_{0}= \frac{3}{4}\left(\frac{M}{Q} \pm \sqrt{\frac{M^2}{Q^2} -\frac{8}{9} } \right)
\end{eqnarray}
corresponding to two photon spheres. In the present paper we consider only the outermost photon sphere given by $\Phi^{-}_{0}$.

\medskip
\noindent

{\bf Case $\frac{Q^2}{M^2}>1$. }

Here we shall use the affine parameter $\lambda= \sqrt{1- \frac{M^2}{Q^2}}{\tilde \lambda}$ with $-\frac{\pi}{2}<\lambda<\frac{\pi}{2}$. As in the previous case we consider the inequalities
(\ref{II1}) and (\ref{II2}) however with different functions $\Gamma$  and $\Omega$, namely

\begin{eqnarray}
\Gamma= \tan(\lambda), \;\;\; \Omega=1+ \Gamma^2=\cos^{-2}(\lambda).
\end{eqnarray}

In the case under consideration the first inequality (\ref{II1}) is equivalent to (\ref{INEQ1}) while the second one to (\ref{INEQ2}).
Hence we have that (\ref{HSEPHI}) is satisfied and therefore $R(\gamma)_{ijk}=0$, i. e. the metric $g_{ij}$ is conformally flat.
The same arguments as in the previous case show that the spatial geometry is spherically symmetric and that  the spacetime  is isometric to the Reissner-Nordstr\"om spacetime with $1<\frac{Q^2}{M^2}\le \frac{9}{8}$. Note that according to eq. (\ref{SOLEPS}) a photon sphere exists only for $\frac{Q^2}{M^2}\le \frac{9}{8}$.
\end{proof}

\section{Discussion}

In the current paper we have proven that  the static asymptotically flat solutions to the Einstein-Maxwell equations with mass $M$ and electric charge $Q$ possessing a
non-extremal photon sphere are isometric to the Reissner-Nordstr\"om spacetime with the same mass and charge subject to the constraint $\frac{Q^2}{M^2}\le\frac{9}{8}$. We have used simple and physically motivated assumptions, namely that the lapse function foliates the spacetime outside of the photon sphere and that the photon sphere is defined with one additional property (compared to the vacuum Einstein case). This property states that the electric field is normal to the photon sphere. It should be
noted that our theorem does not cover the extremal case which is rather subtle and needs more sophisticated techniques.

In seeking mathematical generality one can ask whether the condition that the lapse function foliates the spacetime can be relaxed and one can consider a priory
 non-connected photon spheres along the lines of \cite{Bunting1987}. Our preliminary studies show that this could be done in the case $\frac{Q^2}{M^2}<1$ (the "black hole" case). In the general case however, the price we have to pay in order to relax the mentioned condition is the increase in technical details and complexity. What's more,
 if we relax the condition that the lapse function foliates the spacetime then we have to assume in addition that the spacetime is simply connected. Nevertheless, we intend  to study this problem and the problem of extremal photon spheres in  future publications. \\

\noindent {\bf Acknowledgements:} S. Y. would like to thank  the Research Group Linkage Programme of the Alexander von Humboldt Foundation
for the support. The partial support by the COST Action MP1304  and by Bulgarian NSF grant DFNI T02/6 is also gratefully acknowledged.

\end{document}